\documentclass[10pt,conference]{IEEEtran}
\usepackage{amsmath,paralist,comment,amssymb}
\usepackage{times,color}
\usepackage{euler}
\usepackage{graphicx}
\usepackage{epsfig}
\usepackage{ifthen}

\DeclareMathOperator{\tr}{tr}
 
\DeclareMathOperator{\wt}{wt}

\newcommand{\E}{\mathcal{E}}
\newcommand{\F}{\mathbb{F}}

\newcommand{\X}{{X}}
\newcommand{\Y}{{Y}}
\newcommand{\Z}{{Z}}

\newcommand{\one}{\mathbf{1}}

\newcommand{\bch}{\mathcal{BCH}}
\newcommand{\eg}{{\rm {EG}}}

\newcommand{\Rm}{{\rm {GRM}}}

\newcommand{\nix}[1]{}

\newtheorem{theorem}{Theorem}
\newtheorem{corollary}[theorem]{Corollary}
\newtheorem{lemma}[theorem]{Lemma}

\newtheorem{example}[theorem]{Example}

\begin{document}

\title{Asymmetric Quantum LDPC Codes}

\author{
\authorblockN{Pradeep Kiran Sarvepalli}
\authorblockA{Department of Computer Science\\
Texas A\&M University \\
College Station, TX 77843\\
Email: pradeep@cs.tamu.edu}
\and
\authorblockN{Martin R\"{o}tteler}
\authorblockA{NEC Laboratories America \\
4, Independence Way, Suite 200 \\
Princeton, NJ 08540 \\
Email: mroetteler@nec-labs.com}
\and
\authorblockN{Andreas Klappenecker}
\authorblockA{Department of Computer Science\\
Texas A\&M University \\
College Station, TX 77843\\
Email: klappi@cs.tamu.edu}
}

\maketitle

\begin{abstract}
  Recently, quantum error-correcting codes were proposed that
  capitalize on the fact that many physical error models lead to a
  significant asymmetry between the probabilities for bit flip and
  phase flip errors. An example for a channel which exhibits such
  asymmetry is the combined amplitude damping and dephasing channel,
  where the probabilities of bit flips and phase flips can be related
  to relaxation and dephasing time, respectively. We give systematic
  constructions of asymmetric quantum stabilizer codes that exploit
  this asymmetry. Our approach is based on a CSS construction that
  combines BCH and finite geometry LDPC codes.
\end{abstract}

\section{Introduction}

In many quantum mechanical systems the mechanisms for the occurrence
of bit flip and phase flip errors are quite different. In a recent
paper Ioffe and M\'{e}zard \cite{ioffe07} postulated that quantum
error-correction should take into account this asymmetry. The
main argument given in \cite{ioffe07} is that most of the known
quantum computing devices have relaxation times ($T_1$) that are
around $1-2$ orders of magnitude larger than the corresponding
dephasing times $(T_2)$. In general, relaxation leads to both bit flip
and phase flip errors, whereas dephasing only leads to phase flip
errors.  This large asymmetry between $T_1$ and $T_2$ suggests that
bit flip errors occur less frequently than phase flip errors and a
well designed quantum code would exploit this asymmetry of errors to
provide better performance.  In fact, this observation and its
consequences for quantum error correction, especially quantum fault
tolerance, have prompted investigations from various other researchers
\cite{aliferis07,evans07,stephens07}.

Our goal will be as in \cite{ioffe07} to construct asymmetric quantum
codes for quantum memories and at present we do not consider the issue of
fault tolerance. We first quantitatively justify how noise processes,
characterized in terms of $T_1$ and $T_2$, lead to an asymmetry in the
bit flip and phase flip errors. As a concrete illustration of this we
consider the amplitude damping and dephasing channel. For this
channel we can compute the probabilities of bit flip and phase flips
in closed form. In particular, by giving explicit expressions for the
ratio of these probabilities in terms of the ratio $T_1/T_2$, we show
how the channel asymmetry arises.

After providing the necessary background,
we give two systematic constructions of asymmetric quantum codes
based on BCH and LDPC codes, as an alternative to the randomized 
construction of \cite{ioffe07}. 

\section{Background}
\noindent
Recall that a quantum channel that maps a state $\rho$ to 
\begin{eqnarray}
(1-p_x-p_y-p_z)\rho+ p_x\X\rho\X+p_y\Y\rho\Y + p_z\Z\rho\Z,\label{pauli-channel}
\end{eqnarray}
with 
$  \one= \left[\begin{smallmatrix} 1&0\\0&1\end{smallmatrix}\right]$,
$\X= \left[\begin{smallmatrix} 0&1 \\1&0\end{smallmatrix}\right]$,
$\Y= \left[\begin{smallmatrix} 0&-i \\i&0\end{smallmatrix}\right]$,
$\Z= \left[\begin{smallmatrix} 1&0 \\0&-1\end{smallmatrix}\right]$
is called a \textit{Pauli channel}. For a Pauli channel, one can
respectively determine the probabilities $p_x, p_y, p_z$ that an input
qubit in state $\rho$ is subjected to a Pauli $X$, $Y$, or $Z$ error.

A combined \textit{amplitude damping and dephasing channel}
$\mathcal{E}$ with relaxation time $T_1$ and dephasing time $T_2$ that
acts on a qubit with density matrix $\rho=(\rho_{ij})_{i,j\in
\{0,1\}}$ for a time $t$ yields the density matrix
$$ \mathcal{E}(\rho) = 
\left[
\begin{array}{cc}
1-\rho_{11}e^{-t/T_1} & \rho_{01} e^{-t/T_2} \\
\rho_{10} e^{-t/T_2} 
& \rho_{11}e^{-t/T_1} 
\end{array}
\right]. 
$$ This channel is interesting as it models common decoherence
processes fairly well. We would like to determine the probability
$p_x$, $p_y$, and $p_z$ such that an $X$, $Y$, or $Z$ error occurs in
a combined amplitude damping and dephasing channel. However, it turns
out that this question is not well-posed, since $\mathcal{E}$ is not a
Pauli channel, that is, it cannot be written in the form
(\ref{pauli-channel}). However, we can obtain a Pauli channel
$\mathcal{E}_T$ by a technique called twirling
\cite{ESR+:2007,dankert06}. In our case, the twirling consists of
conjugating the channel $\mathcal{E}$ by Pauli matrices and averaging
over the results. The resulting channel $\mathcal{E}_T$ is called the
Pauli-twirl of $\mathcal{E}$ and is explicitly given by
$$ \mathcal{E}_T(\rho) = \frac{1}{4} \sum_{A \in \{ \one, \X,\Y,\Z\}}
A^\dagger \mathcal{E}(A\rho A^\dagger )A.$$ 

\begin{theorem} Given a combined amplitude damping and dephasing 
channel $\mathcal{E}$ as above, the associated Pauli-twirled channel
is of the form
$$\mathcal{E}_T(\rho) = (1-p_x-p_y-p_z)\rho+ p_x\X\rho\X+p_y\Y\rho\Y +
p_z\Z\rho\Z,$$
where $p_x=p_y=(1-e^{-t/T_1})/4$ and $p_z=1/2-p_x-\frac{1}{2}e^{-t/T_2}$. 
In particular,
$$ \frac{p_z}{p_x} =1+2 \frac{1-e^{t/T_1(1-T_1/T_2)}}{e^{t/T_1}-1}.
$$ 
If $t\ll T_1$, then we can approximate this ratio as $2T_1/T_2-1$.
\end{theorem}
\begin{proof}
  The Kraus operator decomposition \cite{NC:2000} of $\mathcal{E}$ is
\begin{eqnarray}
\mathcal{E}(\rho) = \sum_{k=0}^2A_k\rho A_k^\dagger,\label{eq:krausDecomp}
\end{eqnarray}
where  $A_0= \left[\begin{smallmatrix} 1&0
\\0&\sqrt{1-\lambda-\gamma}\end{smallmatrix}\right]; A_1 = \left[\begin{smallmatrix} 0&
0\\0&\sqrt{\lambda}\end{smallmatrix}\right];A_2 = \left[\begin{smallmatrix} 0&
\sqrt{\gamma}\\0& 0\end{smallmatrix}\right],$
and $\sqrt{1-\gamma-\lambda}= e^{-t/T_2}$,  $1-\gamma=e^{-t/T_1}$. We can rewrite the Kraus
operators $A_i$  as 
\begin{eqnarray*}
A_0=\frac{1+\sqrt{1-\lambda-\gamma}}{2}\one +\frac{1-\sqrt{1-\lambda-\gamma}}{2}\Z,\\
A_1=\frac{\sqrt{\lambda}}{2}\one -\frac{\sqrt{\lambda}}{2}\Z, \quad A_2=\frac{\sqrt{\gamma}}{2}\X -\frac{\sqrt{\gamma}}{2i}\Y.
\end{eqnarray*}
Rewriting $\E(\rho)$ in terms of Pauli matrices 
yields 
\begin{eqnarray}
\E(\rho) &=& \frac{2-\gamma+2\sqrt{1-\lambda-\gamma}}{4} \rho + \frac{\gamma}{4}\X\rho\X +\frac{\gamma}{4}\Y\rho\Y \nonumber \\ 
&+& \frac{2-\gamma-2\sqrt{1-\lambda-\gamma}}{4}\Z\rho\Z \nonumber \\
&-& \frac{\gamma}{4}\one\rho\Z -\frac{\gamma}{4}\Z\rho\one +\frac{\gamma}{4i} \X\rho\Y
-\frac{\gamma}{4i}\Y\rho\X.
\end{eqnarray}
It follows that the Pauli-twirl channel $\mathcal{E}_T$ is of the
claimed form, see~\cite[Lemma 2]{dankert06}.
Computing the ratio $p_z/p_x$ we get
\begin{eqnarray*}
	\frac{p_z}{p_x} &=&\frac{2-\gamma-2\sqrt{1-\lambda-\gamma}}{\gamma} =\frac{1+e^{-t/T_1}-2e^{-t/T_2}}{1-e^{-t/T_1}},\\
	&=&1+2\frac{e^{-t/T_1}-e^{-t/T_2}}{1-e^{-t/T_1}} = 1+2 \frac{1-e^{t/T_1-t/T_2}}{e^{t/T_1}-1} \\
	&=&1+2 \frac{1-e^{t/T_1(1-T_1/T_2)}}{e^{t/T_1}-1}.
\end{eqnarray*}
If $t\ll T_1$, then we can approximate the ratio as $2T_1/T_2-1$, as claimed. 
\end{proof}

Thus, an asymmetry in the $T_1$ and $T_2$ times does translate to an
asymmetry in the occurrence of bit flip and phase flip errors.  Note
that $p_x=p_y$ indicating that the $\Y$ errors are as unlikely as the
$\X$ errors.  We shall refer to the ratio $p_z/p_x$ as the channel
asymmetry and denote this parameter by $A$.

Asymmetric quantum codes use the fact that the phase flip errors are much
more likely than the bit flip errors or the combined bit-phase flip
errors.  Therefore the code has different error correcting capability
for handling different type of errors. We require the code to correct
many phase flip errors but it is not required to handle the same
number of bit flip errors. If we assume a CSS code
\cite{calderbank98}, then we can meaningfully speak of $X$-distance
and $Z$-distance. A CSS stabilizer code that can detect all $X$ errors
up to weight $d_x-1$ is said to have an $X$-distance of $d_x$.
Similarly if it can detect all $Z$ errors upto weight $d_z-1$, then it
is said to have a $Z$-distance of $d_z$. We shall denote such a code
by $[[n,k,d_x/d_z]]_q$ to indicate it is an asymmetric code, see also
\cite{Steane:96} who was the first to use a notation that allowed to
distinguish between $X$- and $Z$-distances.  We could also view this
code as an $[[n,k,\min\{ d_x,d_z\}]]_q$ stabilizer code.  Further
extension of these metrics to an additive non-CSS code is an
interesting problem, but we will not go into the details here.

Recall that in the CSS construction a pair of codes are used, one for
correcting the bit flip errors and the other for correcting the phase
flip errors. Our choice of these codes will be such that the code for
correcting the phase flip errors has a larger distance than the code
for correcting the bit flip errors. We restate the CSS construction in
a form convenient for asymmetric stabilizer codes.

\begin{lemma}[CSS Construction \cite{calderbank98}]\label{lm:css}
  Let $C_x, C_z$ be linear codes over $\F_q^n$ with the parameters
  $[n,k_x]_q$, and $[n,k_z]_q$ respectively. Let $C_x^\perp\subseteq
  C_z$.  Then there exists an $[[n,k_x+k_z-n,d_x/d_z]]_q$ asymmetric
  quantum code, where $d_x=\wt(C_x\setminus C_z^\perp)$ and
  $d_z=\wt(C_z\setminus C_x^\perp)$.
\end{lemma}
If in the above construction $d_x=\wt(C_x)$ and $d_z=\wt(C_z)$, then
we say that the code is pure. 

In the theorem above and elsewhere in this paper  $\F_q$ denotes a finite field with $q$ elements.
We also denote a $q$-ary narrow-sense primitive BCH code of length $n=q^m-1$
and design distance $\delta$ as $\bch(\delta)$. 

\section{Asymmetric Quantum Codes from LDPC Codes}

In \cite{ioffe07}, Ioffe and M\'{e}zard used a combination of BCH and
LDPC codes to construct asymmetric codes. The intuition being that the
stronger LDPC code should be used for correcting the phase flip errors
and the BCH code can be used for the infrequent bit flips. This
essentially reduces to finding a good LDPC code such that the dual of
the LDPC code is contained in the BCH code. They solve this problem by
randomly choosing codewords in the BCH code which are of low weight
(so that they can be used for the parity check matrix of the LDPC
code).  However, this method leaves open how good the resulting LDPC
code is.  For instance, the degree profiles of the resulting code are
not regular and there is little control over the final degree profiles
of the code.  Furthermore, it is not apparent what ensemble or degree
profiles one will use to analyze the code.

We propose an alternate scheme that uses LDPC codes to construct asymmetric
stabilizer codes. We propose two families of quantum codes based on LDPC codes.  In the
first case we use LDPC codes for both the $X$ and $Z$ channel while in the
second construction we will use a combination of BCH and LDPC codes. But 
first, we will need the following facts about generalized Reed-Muller 
codes and finite geometry LDPC codes.

\subsection{Finite Geometry LDPC Codes (\cite{kou01,tang05})} 
Let us denote by $\eg(m,p^s)$ the Euclidean finite geometry over
$\F_{p^s}$ consisting of $p^{ms}$ points. For our purposes it suffices
to use the fact that this geometry is equivalent to the vector space
$\F_{p^s}^m$.  A $\mu$-dimensional subspace of $\F_{p^s}^m$ or its
coset is called a \textsl{$\mu$-flat}. Assume that $0\leq \mu_1<
\mu_2\leq m$.  Then we denote by $N_{\eg}(\mu_2,\mu_1,s,p)$ the number
of $\mu_1$-flats in a $\mu_2$-flat and by $A_{\eg}(m,\mu_2,\mu_1,
s,p)$, the number of $\mu_2$-flats that contain a given $\mu_1$-flat.
These are given by (see \cite{tang05})
\begin{eqnarray}
N_{\eg}(\mu_2,\mu_1,s,p)&=& q^{(\mu_2-\mu_1)}\prod_{i=1}^{\mu_1}\frac{q^{\mu_2-i+1}-1}{q^{\mu_1-i+1}-1},\label{eq:nEG}\\
A_{\eg}(m,\mu_2,\mu_1,s,p)&=& \prod_{i=\mu_1+1}^{\mu_2}\frac{q^{m-i+1}-1}{q^{\mu_2-i+1}-1}, \label{eq:aEG}
\end{eqnarray}
where $q=p^s$.  Index all the $\mu_1$-flats from $i=1$ to
$n=N_{\eg}(m,\mu_1,s,p)$ as $F_i$.  Let $F$ be a $\mu_2$-flat in
$\eg(m,p^s)$.  Then we can associate an incidence vector to $F$ with
respect to the $\mu_1$ flats as follows.
$$
\mathbf{i}_F = \left\{i_j\mid \begin{array}{cl}i_j =1 &\mbox{ if $F_j$ is contained in } F \\
i_j=0& \text{otherwise.}\end{array}
\right\}.
$$
Index the $\mu_2$-flats from $j=1$ to $J=N_{\eg}(m,\mu_2,s,p)$.
Construct the $J\times n $ matrix $H_{\eg}^{(1)}(m,\mu_2,\mu_1,s,p)$
whose rows are the incidence vectors of all the $\mu_2$-flats with
respect to the $\mu_1$-flats. This matrix is also referred to as the
incidence matrix.  Then the type-I Euclidean geometry code from
$\mu_2$-flats and $\mu_1$-flats is defined to be the null space,
i.\,e., Euclidean dual code) of the $\F_p$-linear span of
$H_{\eg}^{(1)}(m,\mu_2,\mu_1,s,p)$. This is denoted as
$C_{\eg}^{(1)}(m,\mu_2,\mu_1,s,p)$.  Let
$H_{\eg}^{(2)}(m,\mu_2,\mu_1,s,p) =
H_{\eg}^{(1)}(m,\mu_2,\mu_1,s,p)^t.$ Then the type-II Euclidean
geometry code $C_{\eg}^{(2)}(m,\mu_2,\mu_1,s,p)$ is defined to be the
null space of $H_{\eg}^{(2)}(m,\mu_2,\mu_1,s,p)$. Let us now consider
the $\mu_2$-flats and $\mu_1$-flats that do not contain the origin of
$\eg(m,p^s)$. Now form the incidence matrix of the $\mu_2$-flats with
respect to the $\mu_1$-flats not containing the origin.  The null
space of this incidence matrix gives us a quasi-cyclic code in
general, which we denote by $C_{\eg,c}^{(1)}(m,\mu_2,\mu_1,s,p)$, see
\cite{tang05}.

\subsection{Generalized Reed-Muller codes (\cite{kasami68})} 
Let $\alpha$ be a primitive element in $\F_{q^{m}}$.  The cyclic
generalized Reed-Muller code of length $q^m-1$ and order $\nu$ is
defined as the cyclic code with the generator polynomial whose roots
$\alpha^j$ satisfy $0<j\leq m(q-1)-\nu-1$. The generalized Reed-Muller
code is the singly extended code of length $q^m$.  It is denoted as
$\Rm_q(\nu,m)$. The dual of a GRM code is also a GRM code
\cite{assmus98,blahut03,kasami68}.  It is known that
\begin{eqnarray}
\Rm_q(\nu,m)^\perp = \Rm_q(\nu^\perp,m),
\end{eqnarray}
where $\nu^\perp = m(q-1)-1-\nu$.

Let $C$ be a linear code over $\F_{q^s}^n$. Then we define
$C|_{\F_q}$, the \textsl{subfield subcode} of $C$ over $\F_q^n$ as the
codewords of $C$ which are entirely in $\F_q^n$, (see
\cite[pages~116-120]{huffman03}). Formally this can be expressed as
\begin{eqnarray}
C|_{\F_q} = \{c\in C\mid c\in \F_q^n \}.
 \end{eqnarray}
Let $C\subseteq \F_{q^l}^n$. The the \textsl{trace code} of $C$ over $\F_q$ is defined  as
\begin{eqnarray}
\tr_{q^l/q}(C) = \{\tr_{q^l/q}(c) \mid c\in C\}.
\end{eqnarray}
There are interesting relations between the trace code and the
subfield subcode. One of which is the following result which we will
need later.
\begin{lemma}\label{lm:inclusion}
Let $C\subseteq \F_{q^l}^n$. Then $C|_{\F_q}$, the subfield subcode of $C$
is contained in $\tr_{q^l/q}(C)$, the trace code of $C$. In other words
$$ C|_{\F_q}\subseteq \tr_{q^l/q}(C).$$
\end{lemma}
\begin{proof}
  Let $c\in C|_{\F_q} \subseteq \F_q^n$ and $\alpha \in \F_{q^l}$.
  Then $\tr_{q^l/q}(\alpha c)= c\tr_{q^l/q}(\alpha)$ as $c\in \F_q^n$.
  Since trace is a surjective form, there exists some $\alpha\in
  \F_{q^l}$, such that $\tr_{q^l/q}(\alpha)=1$. This implies that
  $c\in \tr_{q^l/q}(C)$. Since $c$ is an arbitrary element in
  $C|_{\F_q}$ it follows that $ C|_{\F_q}\subseteq \tr_{q^l/q}(C)$.
\end{proof}

Let $q=p^s$, then the Euclidean geometry code of order $r$ over
$\eg(m,p^s)$ is defined as the dual of the subfield subcode of
$\Rm_{q}((q-1)(m-r-1),m)$, \cite[page~448]{blahut03}. The type-I LDPC
code $C_{\eg}^{(1)}(m,\mu,0,s,p)$ code is an Euclidean geometry code
of order $\mu-1$ over $\eg(m,p^s)$, see \cite{tang05}.
Hence its dual is the subfield subcode of $\Rm_q((q-1)(m-\mu),m)$
code.  In other words,
\begin{eqnarray}
C_{\eg}^{(1)}(m,\mu,0,s,p)^\perp = 
\Rm_q((q-1)(m-\mu),m)|_{\F_p}.\label{eq:fgSubfiledCode}
\end{eqnarray}
Further, Delsarte's theorem \cite{delsarte75} tells us that 
\begin{eqnarray*}
C_{\eg}^{(1)}(m,\mu,0,s,p)&=&\Rm_q((q-1)(m-\mu),m)|_{\F_p}^\perp ,\\
&=&\tr_{q/p}\left( \Rm_q((q-1)(m-\mu),m)^\perp\right)\\
&=&  \tr_{q/p}(\Rm_q(\mu(q-1)-1,m)).
\end{eqnarray*}
Hence,  $C_{\eg}^{(1)}(m,\mu,0,s,p)$ code can also be  related to $\Rm_q(\mu(q-1)-1,m)$ as 
\begin{eqnarray}
  C_{\eg}^{(1)}(m,\mu,0,s,p)=\tr_{q/p}(\Rm_q(\mu(q-1)-1),m).\label{eq:fgTraceCode}
\end{eqnarray}

\subsection{New families of asymmetric quantum codes}
With the previous preparation we are now ready to construct asymmetric quantum
codes from finite geometry LDPC codes. 

\begin{theorem}[Asymmetric EG LDPC Codes]
  Let $p$ be a prime, with $q=p^s$ and $s\geq 1,m\geq 2$. Let $1<
  \mu_z <m$ and $m-\mu_z+1 \leq \mu_x<m$.  Then there exists an
$$
[[p^{ms},k_x+k_z-p^{ms} , d_x/d_z]]_p
$$
asymmetric EG LDPC code, where 
$$k_x=\dim C_{\eg}^{(1)}(m,\mu_x,0,s,p); \quad  k_z=\dim C_{\eg}^{(1)}(m,\mu_z,0,s,p).$$ For the distances 
$d_x\geq A_{\eg}(m,\mu_x,\mu_x-1,s,p)+1$ and $d_z\geq
A_{\eg}(m,\mu_z,\mu_z-1,s,p)+1$ hold.
\end{theorem}
\begin{proof}
  Let $C_z=C_{\eg}^{(1)}(m,\mu_z,0,s,p)$.  Then from
  equation~(\ref{eq:fgTraceCode}) we have
\begin{eqnarray*}
C_z &=& \tr_{q/p}(\Rm_q(\mu_z(q-1)-1,m).
\end{eqnarray*}
By Lemma~\ref{lm:inclusion} we know that 
\begin{eqnarray*}
C_z &\supseteq&  \Rm_q(\mu_z(q-1)-1,m)|_{\F_p},\\
C_z&\supseteq & \Rm_q((q-1)(m-(m-\mu_z+1)),m)|_{\F_p},
\end{eqnarray*}
where the last inclusion follows from the nesting property of the
generalized Reed-Muller codes. For any order $\mu_x$ such that
$m-\mu_z+1\leq \mu_x <m$, let $C_x=C_{\eg}^{(1)}(m,\mu_x,0,s,p)$. Then
$C_x$ is an LDPC code whose dual $C_x^\perp =
\Rm_q((q-1)(m-\mu_x),m)|_{\F_p}$ is contained in $C_z$. Thus we can
use Lemma~\ref{lm:css} to form an asymmetric code with the parameters
$$ 
[[p^{ms},k_x+k_z-p^{ms} , d_x/d_z]]_p
$$
The distance of $C_z$ and $C_x$ are at lower bounded as 
$d_x \geq A_{\eg}(m,\mu_x,\mu_x-1,s,p)+1$ 
and $d_z \geq A_{\eg}(m,\mu_z,\mu_z-1,s,p)+1$  (see \cite{tang05}).
\end{proof}

In the construction just proposed, we should choose $C_z$ to be a
stronger code compared to $C_x$. We have given the construction over 
a nonbinary alphabet even though the case $p=2$ might be of particular 
interest.

Our next construction makes use of the cyclic finite geometry
codes.  Our goal will be to find a small BCH code whose dual is
contained in a cyclic Euclidean geometry LDPC code. For solving this
problem we need to know the cyclic structure of
$C_{\eg,c}^{(1)}(m,\mu,0,s,p)$. Let $\alpha$ be a primitive element in
$\F_{p^{ms}}$. Then the roots of generator polynomial of
$C_{\eg,c}^{(1)}(m,\mu,0,s,p)$ are given by
\cite[Theorem~6]{kasami71}, see also \cite{kasami68b,lin04}. Now, 
$$
Z= \{\alpha^h\mid 0< \max_{0\leq l<s} W_{p^s}(h p^l) \leq
(p^s-1)(m-\mu) \},
$$
where $W_{q}(h)$ is the $q$-ary weight of $h = h_0+h_1q+\cdots +
h_kq^{k-1}$, i.\,e., $W_q(h)=\sum h_i$.  The finite geometry code
$C_{\eg,c}^{(1)}(m,\mu,0,s,p)$ is actually an $(\mu-1,p^s)$ Euclidean
geometry code.  The roots of the generator polynomial of the dual code
are given by
$$
Z^{\perp}= \{\alpha^h \mid \min_{0\leq l<s} W_{p^s}(hp^l) < \mu(p^s-1) \}.
$$
In fact, the dual code is the even-like subcode of a primitive
polynomial code of length $p^{ms}-1$ over $\F_p$ and order $m-\mu$,
whose generator polynomial, by \cite[Theorem~6]{kasami68b}, has the
roots
$$
Z_p= \{\alpha^h \mid 0< \min_{0\leq l<s} W_{p^s}(hp^l) < \mu(p^s-1) \}.
$$
Thus $Z^\perp = Z_p\cup \{ 0\}$.  Now by \cite[Theorem~11]{kasami68b},
$Z_p$ and therefore $Z^\perp$ contain the sequence of consecutive
roots, $\alpha,\alpha^2,\ldots, \alpha^{\delta_0-1}$, where
$\delta_0=(R+1)p^{Qs}-1$ and $m(p^s-1)-(m-\mu)(p^s-1) = Q(p^s-1)+R$.
Simplifying, we see that $R=0$ and $Q=\mu$ giving $\delta_0 = p^{\mu s}-1$.
It follows that 
\begin{eqnarray*}
C_{\eg,c}^{(1)}(m,\mu,0,s,p)^\perp &=& \Rm_q(m,(q-1)(m-\mu))|_{\F_p} \\
&\subseteq& \bch(\delta_0).
\end{eqnarray*}

Thus we have solved the problem of construction of the asymmetric
stabilizer codes in a dual fashion to that of \cite{ioffe07}. Instead of
finding an LDPC code whose parity check matrix is contained in a given
BCH code, we have found a BCH code whose parity check matrix is
contained in a given finite geometry LDPC code. This gives us the
following result.
\begin{theorem}[Asymmetric BCH-LDPC stabilizer codes]
  Let $C_z=C_{\eg,c}^{(1)}(m,\mu,0,s,p)$ and $\delta\leq \delta_0=p^{\mu s}-1$.  
  Let $n=p^{ms}-1$ and $C_x=\bch(\delta)\subseteq \F_p^n$. Then there exists an
$$[[n,k_x+k_z-n,d_x/d_z]]_p$$
asymmetric stabilizer code where $d_z\geq A_{\eg}(m,\mu,\mu-1,s,p)$,
$d_x\geq \delta$  and $k_x=\dim C_x$, $k_z=\dim C_z$.
\end{theorem}
Perhaps an example will be helpful at this juncture.
\begin{example}
  Let $m=s=p=2$ and $\mu=1$. Then $C_{\eg,c}^{(1)}(2,1,0,2,2)$ is a
  cyclic code whose generator polynomial has roots given by
\begin{eqnarray*}
  Z&=&\{\alpha^h|0<\max_{0\leq l<2 }W_{2^2}(2^lh) \leq (m-\mu)(p^s-1)=3 \}\\
  &=&\{  \alpha^1, \alpha^2, \alpha^3, \alpha^4, \alpha^6, \alpha^8, \alpha^9, \alpha^{12} \}.
\end{eqnarray*}
As there are 4 consecutive roots and $|Z|=8$, it defines a $[15,7,\geq
5]$ code.  The roots of the generator polynomial of the dual code are
given by
\begin{eqnarray*}
Z^\perp&=& \{\alpha^h|0<\min_{0\leq l<2 }W_{2^2}(2^lh) \leq \mu(p^s-1)=(2^2-1) \}\\
&=&\{ \alpha^0, \alpha^1, \alpha^2, \alpha^4, \alpha^5, \alpha^8, \alpha^{10} \}.
\end{eqnarray*}
We see that $Z^\perp$ has two consecutive roots excluding $1$,
therefore the dual code is contained in a narrowsense BCH code with
design distance 3. Note that $p^{\mu s}-1=3$.  Thus we can choose
$C_x=\bch(3)$ and $C_z=C_{\eg,c}^{(1)}(2,1,0,2,2)$ and apply
Lemma~\ref{lm:css} to construct a $[[15,3,3/5]]_2$ asymmetric code.
\end{example}

We can also state the above construction as in \cite{ioffe07}, that is
given a primitive BCH code of design distance $\delta$, find an
LDPC code whose dual is contained in it. It must be pointed out that
in case of asymmetric codes derived from LDPC codes, the asymmetry
factor $d_x/d_z$ is not as indicative of the code performance as in
the case of bounded distance decoders.  For $m=p=2$, we can derive
explicit relations for the parameters of the codes.
\begin{corollary}\label{co:2dCyclic}
  Let $C=C_{\eg,c}^{(1)}(2,1,0,s,2)$ and $\delta =2t+1\leq 2^{s}-1$.
  Then there exists an
$$[[2^{2s}-1,2^{2s}-3^s -s(\delta-1),\delta/2^s+1]]_2 $$ asymmetric stabilizer code. 
\end{corollary}
\begin{proof}
  The parameters of $C$ are $[2^{2s}-1,2^{2s}-3^s,2^{s}+1]_2$, see
  \cite{lin04}. Since $C^\perp $ is contained in a BCH code of length
  $2^{2s}-1$ whose design distance $\delta\leq 2^s-1$, we can compute
  the dimension of the BCH code as $2^{2s}-1 -s(\delta-1)$, see
  \cite[Corollary~8]{macwilliams77}. By Lemma~\ref{lm:css} the quantum
  code has the dimension $2^{2s}-3^s -s(\delta-1)$.
\end{proof}

\begin{example}
  For $m=p=2$ and $s=4$ we can obtain a $[255,175,17]$ LDPC code. We
  can choose any BCH code with design distance $\delta \leq 2^4-1=15$
  to construct an asymmetric code. Table~\ref{tab:bchLdpcAqecc} lists possible codes.
\begin{table}[htb]
\begin{center}
\caption{Asymmetric BCH-LDPC stabilizer codes}
\label{tab:bchLdpcAqecc}
\begin{tabular}{c|c|c|l|c}
$s$&$\delta$  & Code &   Asymmetry &Rate\\
&  & $[[n,k,d_x/d_z]]_2$ &    $d_z/d_x$ &\\\hline
4&  15 & $[[255,119,15/17]]_2$   &   $\approx 1$&0.467\\
4&13  & $[[255,127,13/17]]_2 $&  $\approx 1.25$&0.498\\
4&11  & $[[255,135,11/17]]_2 $&  $\approx 1.5$&0.529\\
4&9  & $[[255,143,9/17]]_2 $&  $\approx 2$ &0.561\\
4&7  & $[[255,151,7/17]]_2 $&  $\approx 2.5$&0.592\\
4&5  & $[[255,159,5/17]]_2 $&  $\approx 3$&0.624\\
4&3  & $[[255,167,3/17]]_2 $&  $\approx 6$&0.655\\
\end{tabular}
\end{center}
\end{table}
\end{example}

\section{Performance Results}
We now study the  performance of the codes constructed
in the previous section. Due to space constraints the discussion will
be rather brief, but more details will be supplied in a forthcoming
paper. We assume that the overall probability of error in the channel is
given by $p$, while the individual probabilities of $X$, $Y$, and $Z$
errors are $p_x=p/(A+2)$, $p_y=p/(A+2)$ and $p_z=pA/(A+2)$
respectively. The exact performance would require us to simulate a
$4$-ary channel and also account for the fact that some errors can be
estimated modulo the stabilizer. However, we do not account for this
and in that sense these results provide an upper bound on the actual
error rates.  
The 4-ary channel can be modeled as two
binary symmetric channels -- one modeling the bit flip channel and the
other the phase flip channel. For exact performance, these two
channels should be dependent, however, a good approximation is to
model the channel as two independent BSCs with cross over
probabilities $p_x+p_y=2p/(A+2)$ and $p_y+p_z=p(A+1)/(A+2)$. In this
case the overall error rate in the quantum channel is the sum of the
error rates in the two BSCs. While this approach is going to slightly
overestimate the error rates, nonetheless it is useful and has been
used before \cite{mackay04}.  Since the $X$-channel uses a BCH code
and decoded using a bounded distance decoder, we can just compute
$P_e^x$ the $X$ error rate, in closed form. 
The error rate in the Z channel, $P_e^z$  is obtained through simulations. The overall
error rate is 
$$P_e=1-(1-P_e^x)(1-P_e^z)=P_e^x+P_e^z-P_e^xP_e^z\approx P_e^x+P_e^z.$$ 

The LDPC code was decoded using the
hard decision bit flipping algorithm given in \cite{kou01}. 
The maximum number of iterations for decoding is set to 50.
In Figure~\ref{fig:fg255d5} we see the performance of
$[[255,159,5/17]]$ as the channel asymmetry is varied from 1 to 100.
We can clearly see the improvement as the channel asymmetry increases.
\begin{figure}[htb]
\begin{center}
\includegraphics[width=9cm]{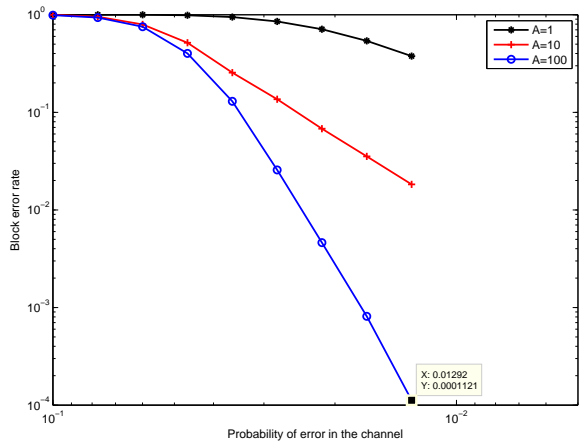}
\caption{Performance of  $[[255,159,5/17]]$ code for $A=1,10,100$}\label{fig:fg255d5}
\end{center}
\end{figure}

The question naturally raises how do these codes compare with the
codes proposed in \cite{ioffe07}. Strictly speaking both constructions
have regimes where they can perform better than the other. But it
appears that the algebraically constructed asymmetric codes have the
following benefits with respect to the randomly constructed ones of
\cite{ioffe07}.
\begin{itemize}
\item They give comparable performance and higher data rates with
  shorter lengths.
\item The benefits of classical algebraic LDPC codes are inherited, giving for instance
lower error floors compared to the random constructions.
\item The code construction is systematic.
\end{itemize}
Our codes also offer flexibility in the rate and performance of the
code because we can choose many possible BCH codes for a given finite
geometry LDPC code or vice versa.  The flip side however is that the
codes given here have higher complexity of decoding.

\section*{Acknowledgment} The authors would like to thank Marcus
Silva for many useful discussions and for proposing the combined
amplitude damping and dephasing channel which is our main motivating
example.

\end{document}